\documentclass{article}
\usepackage[english]{babel}
\usepackage{hyperref,amsmath,amsfonts,amssymb,amsthm,mathrsfs,textcomp}

\title{Hartree Corrections in a Mean-field Limit for Fermions with Coulomb Interaction}

\author{S\"oren Petrat\footnote{Department of Physics, Jadwin Hall, Princeton University, Washington Road, Princeton, NJ 08544, USA. E-mail: {\tt spetrat@princeton.edu}}}

\date{March 30, 2017}

\allowdisplaybreaks[1]

\theoremstyle{plain}\newtheorem{theorem}{Theorem}[section]
\theoremstyle{plain}\newtheorem{lemma}[theorem]{Lemma}

\newcommand{\scp}[2]{\Big\langle #1 , #2 \Big\rangle}
\newcommand{\lilscp}[2]{\langle #1 , #2 \rangle}
\newcommand{\bigscp}[2]{\bigg\langle #1 , #2 \bigg\rangle}
\newcommand{\ketbr}[1]{| #1 \rangle \langle #1 |}
\newcommand{\norm}[2][]{\left|\left| #2 \right|\right|_{#1}}
\renewcommand{\Im}{\mathrm{Im}}
\newcommand{\RRR}{\mathbb{R}}
\newcommand{\ZZZ}{\mathbb{Z}}
\newcommand{\tr}{\mathrm{tr}}
\newcommand{\free}{\mathrm{free}}
\newcommand{\eff}{\mathrm{eff}}
\newcommand{\phieff}{\widetilde{\varphi}}
\newcommand{\rhoeff}{\widetilde{\rho}}
\newcommand{\qeff}{\widetilde{q}}
\newcommand{\peff}{\widetilde{p}}
\newcommand{\alphaeff}{\widetilde{\alpha}}
\newcommand{\absatz}{\vspace{0.4cm}}

\newcounter{remarks}
\begin{document}

\maketitle

\begin{abstract}
We consider the many-body dynamics of fermions with Coulomb interaction in a mean-field scaling limit where the kinetic and potential energy are of the same order for large particle numbers. In the considered limit the spatial variation of the mean-field is small. We prove two results about this scaling limit. First, due to the small variation, i.e., small forces, we show that the many-body dynamics can be approximated by the free dynamics with an appropriate phase factor with the conjectured optimal error term. Second, we show that the Hartree dynamics gives a better approximation with a smaller error term. In this sense, assuming that the error term in the first result is optimal, we derive the Hartree equations from the many-body dynamics with Coulomb interaction in a mean-field scaling limit.
\end{abstract}

\noindent\textbf{MSC class:} 35Q40, 35Q55, 81Q05, 82C10


\section{Introduction}

The quantum many-body dynamics of $N$ non-relativistic spinless particles in three dimensions is described by a wave function $\psi^t \in L^2(\RRR^{3N})$, the space of complex square-integrable functions. The time evolution is governed by the Schr\"odinger equation
\begin{equation}\label{Schroedinger_intro}
i \partial_t \psi^t = H \psi^t,
\end{equation}
where the Hamiltonian $H$ is a self-adjoint operator on a domain dense in $L^2(\RRR^{3N})$. Here, we consider units where $\hbar = 1 = 2m$. Given initial conditions $\psi^0 \in L^2(\RRR^{3N})$, by Stone's theorem, the wave function at time $t$ is given by $\psi^t = e^{-iHt} \psi^0$. The Hamiltonian $H$ is usually of the form
\begin{equation}\label{Hamiltonian}
H = \sum_{i=1}^N \Big( -\Delta_i + W(x_i) \Big) + \lambda_N \sum_{1\leq i<j \leq N} v(x_i-x_j),
\end{equation}
where $\Delta_i$ denotes the Laplacian acting on the $i$-th variable, $W:\RRR^3\to\RRR$ describes an external field, $v:\RRR^3\to\RRR$ with $v(x) = v(-x)$ is a pair interaction potential, and $\lambda_N \in \RRR$ is called a coupling constant. There are two different kinds of wave functions, those describing bosons and those describing fermions. Bosonic wave functions $\psi^{B}$ are symmetric, i.e.,
\begin{equation}
\psi^{B}(x_1,\ldots,x_N) = \psi^{B}(x_{\sigma(1)}, \ldots, x_{\sigma(N)})
\end{equation}
for any permutation $\sigma$, whereas fermionic wave functions $\psi^{F}$ are antisymmetric,
\begin{equation}
\psi^{F}(x_1,\ldots,x_N) = (-1)^{\sigma} \psi^{F}(x_{\sigma(1)}, \ldots, x_{\sigma(N)}),
\end{equation}
where $(-1)^{\sigma}$ denotes the sign of the permutation $\sigma$. For Hamiltonians with pair interaction with $v(x) = v(-x)$, the symmetry properties of the wave function are preserved by the time evolution \eqref{Schroedinger_intro}, i.e., bosons remain bosons and fermions remain fermions.

For studying a dynamical problem in quantum physics, one would like to know about properties of the wave function $\psi^t$, given initial conditions $\psi^0$. Ideally, one would like to solve the Schr\"odinger equation \eqref{Schroedinger_intro} exactly, but this is usually practically impossible for large particle number $N$, due to the pair interaction $v$ and the resulting complexity of $\psi^t$ on the very high dimensional space $L^2(\RRR^{3N})$. Note that the high dimensionality also makes numerical solutions practically impossible. Moreover, for large particle number $N$, one might be rather interested in statistical properties of $\psi^t$, e.g., certain averages that tell us about the typical behavior of particles. Thus, it is very interesting to study approximations to $\psi^t$, especially for large particle number $N$. In this article, we study one such approximation for fermions, namely Hartree theory, for certain suitable initial conditions and coupling constants. The emphasis is on Coulomb interaction $v(x) = \pm |x|^{-1}$, which is physically very relevant since it describes, e.g., the interaction of electrons ($+|x|^{-1}$) or particles with gravitational interaction ($-|x|^{-1}$).

We consider the approximation of $\psi$ by the most simple antisymmetric state, that is, an antisymmetric product state defined by
\begin{equation}
\left(\bigwedge_{j=1}^N \varphi_j\right)(x_1,\ldots,x_N) = \frac{1}{\sqrt{N!}} \sum_{\sigma \in S_N} (-1)^{\sigma} \prod_{i=1}^N \varphi_{\sigma(i)}(x_i),
\end{equation}
where $\{\varphi_j\}_{j=1,\ldots,N}$ is a family of orthonormal one-body wave functions $\in L^2(\RRR^3)$, and $S_N$ denotes the symmetric group. We often abbreviate $\bigwedge_{j=1}^N \varphi_j = \bigwedge \varphi_j$. When the initial condition $\psi^0$ is approximately given by $\bigwedge \varphi_j^0$, we would like to approximate $\psi^t$ by a wave function $\bigwedge \varphi_j^t$ with suitable $\varphi_j^t$. For expressing the degree of approximation in a macroscopic sense, it is useful to introduce reduced density matrices. The reduced one-particle density matrix $\gamma_{\psi}:L^2(\RRR^3)\to L^2(\RRR^3)$ of a wave function $\psi$ is defined by its kernel
\begin{equation}
\gamma_{\psi}(x,y) = \int dx_2\ldots dx_N \psi(x,x_2,\ldots,x_N) \overline{\psi(y,x_2,\ldots,x_N)},
\end{equation}
where $\overline{f}$ denotes the complex conjugate of $f$. Here, we have normalized $\gamma_{\psi}$ such that $\tr \gamma_{\psi}=\norm{\psi}^2=1$, where $\tr$ denotes the trace. Note that $\gamma_{\wedge\varphi_j}(x,y) = N^{-1} \sum_{j=1}^N \varphi_j(x) \overline{\varphi_j(y)}$. In the main results, we provide estimates for the distance of reduced density matrices in trace norm, i.e., for $\norm[\tr]{\gamma_{\psi^t} - \gamma_{\wedge\varphi_j^t}}$. By 
\begin{equation}
\bigscp{\psi^t}{A_1\psi^t} - \bigscp{\bigwedge_{j=1}^N\varphi_j^t}{A_1 \bigwedge_{j=1}^N\varphi_j^t} = \tr \Big[ A \big( \gamma_{\psi^t} - \gamma_{\wedge\varphi_j^t} \big) \Big] \leq \norm{A} \norm[\tr]{\gamma_{\psi^t} - \gamma_{\wedge\varphi_j^t}},
\end{equation}
where $A_1:L^2(\RRR^{3N}) \to L^2(\RRR^{3N})$ is the one-body operator $A:L^2(\RRR^3)\to L^2(\RRR^3)$ acting only on the first variable, this allows to control the expectation values of bounded one-particle observables ($\lilscp{\cdot}{\cdot}$ denotes the scalar product on $L^2$). Let us now, for simplicity of the presentation, disregard external fields $W(x)$ (the main result also holds for regular enough external fields) and consider repulsive interaction $v(x) = |x|^{-1}$, i.e., the Schr\"odinger equation is
\begin{equation}\label{Schr_lambda}
i\partial_t \psi^t = \left( -\sum_{i=1}^N \Delta_i + \lambda_N \sum_{1\leq i<j \leq N} |x_i-x_j|^{-1} \right) \psi^t.
\end{equation}
We consider initial conditions with kinetic energy $\lilscp{\psi^0}{ \sum_{j} (-\Delta_j)\psi^0}$ of order $N$. This means that the kinetic energy is an extensive quantity here. As a consequence, the wave function is supported at least on a volume of order $N$ (i.e., the density is at most of order $1$), e.g., think of a box with diameter $N^{1/3}$. Note that this is very different from a bosonic wave function, which could be supported on a volume of order one instead. If we now choose $\lambda_N = N^{-2/3}$, then also the interaction energy $\lilscp{\psi^0}{ \lambda_N \sum_{i<j} |x_i-x_j|^{-1}\psi^0}$ is of order $N$. This is due to the large volume. For example, think of a wave function with constant absolute value in a ball of radius $N^{1/3}$. Then the interaction energy is $N^{-2/3} N \int_{0}^{N^{1/3}} r^{-1} r^2 dr \propto N$. Therefore, it should be possible to use mean-field theory to approximate $\psi^t$. For fermions, this means that the one-particle wave functions $\varphi_j^t$ are solutions to the fermionic Hartree equations
\begin{equation}\label{Hartree_intro}
i \partial_t \varphi_j^t = \left( -\Delta + N^{-2/3} \big(|\cdot|^{-1}*\rho^t\big) \right) \varphi_j^t,
\end{equation}
where $\rho^t = \sum_{j=1}^N |\varphi_j^t|^2$, and $*$ denotes convolution. Here, the two-body interaction has been replaced by the mean-field one-body term $|\cdot|^{-1}*\rho^t$, which makes the equation \eqref{Hartree_intro} nonlinear. Note that we have not included the so-called exchange term which would lead to the Hartree-Fock equations, since it is subleading in the considered scaling limit. In order to show that the Hartree dynamics approximates the Schr\"odinger dynamics well, we would like to prove that for all times $t$,
\begin{equation}
\norm[\tr]{\gamma_{\psi^0} - \gamma_{\wedge\varphi_j^0}} \to 0 ~~~\Rightarrow~~~ \norm[\tr]{\gamma_{\psi^t} - \gamma_{\wedge\varphi_j^t}} \to 0
\end{equation}
for $N\to\infty$ and all $t>0$. However, the considered scaling limit is a bit subtle. We noted that for large $N$ the kinetic and interaction energies are of the same order, but this does not necessarily mean that both terms are competing in the dynamics. Indeed, the mean-field from Equation~\eqref{Hartree_intro} varies only very slowly on spatial scales of order one. In order to see this on a heuristic level, think of a bounded density $\rho_b$ that is supported on a ball of radius $N^{1/3}$. Then
\begin{equation}\label{heuristic_nabla_v_star_rho}
\nabla_x N^{-2/3} \big(|\cdot|^{-1}*\rho_b\big)(x) = N^{-2/3} \big((\nabla|\cdot|^{-1})*\rho_b\big)(x) \leq C N^{-2/3} \int_0^{N^{1/3}} r^{-2} r^2 dr \propto N^{-1/3}.
\end{equation}
In other words, the average forces in the system are only of order $N^{-1/3}$. Thus, it should be possible to approximate the time evolution of $\varphi_j^t$ to leading order by the free evolution with an appropriate phase factor to account for the large potential energy, i.e., by
\begin{equation}\label{free_with_phase_intro}
\phieff_j^t := e^{-i\Phi^t(x)} \varphi_j^{\free,t}, ~~\text{with}~~ \Phi^t(x) = N^{-2/3} \int_0^t ds\, \Big( |\cdot|^{-1}*\sum_{j=1}^N|\varphi_j^{\free,s}|^2 \Big)(x),
\end{equation}
where
\begin{equation}\label{free_intro}
i \partial_t \varphi_j^{\free,t} = -\Delta \varphi_j^{\free,t}.
\end{equation}
The free dynamics with the phase factor $\Phi^t(x)$ is the lowest order approximation to the Hartree equation \eqref{Hartree_intro}, using that the gradient of the mean-field potential is very small (similar to $e^{-i(-\Delta + V)t}\varphi \approx e^{-iVt}e^{-i(-\Delta)t}\varphi$ when $[\Delta,V]$ is small). However, we would still like to see the effect of the interaction in the approximation for $\psi^t$. Therefore, we need to take a closer look at the error terms in the convergence of the reduced density matrices, i.e., the convergence rates. For the convergence to the free dynamics with phase factor \eqref{free_with_phase_intro} and suitable initial conditions, we can at best expect
\begin{equation}\label{intro_free}
\norm[\tr]{\gamma_{\psi^t} - \gamma_{\wedge\phieff_j^t}} \leq C(t) N^{-1/3}.
\end{equation}
Heuristically, this is so because according to \eqref{heuristic_nabla_v_star_rho} each particle feels a force of order $N^{-1/3}$. This means that on average only one in $N^{1/3}$ particles feels an interaction of order one, i.e., only a fraction $N^{2/3}$ out of $N$ particles interact with order one. Thus, the average deviation from the free dynamics with phase factor should be of order $N^{2/3}/N$. Indeed, we prove this rigorously as an upper bound in Theorem~\ref{thm:main_result1}. On the other hand, for the convergence to the fermionic Hartree equations \eqref{Hartree_intro}, we can prove a better convergence rate,
\begin{equation}
\norm[\tr]{\gamma_{\psi^t} - \gamma_{\wedge\varphi_j^t}} \leq C(t) N^{-1/2},
\end{equation}
see Theorem~\ref{thm:main_result2}. Thus, if we assume that the error term in \eqref{intro_free} is optimal, then the fermionic Hartree equations indeed provide a subleading correction to the dynamics of fermions with Coulomb interaction in a mean-field scaling limit. Concerning the convergence rates for singular interactions in the considered scaling limit, this work generalizes \cite{froehlich:2011} in that it provides a convergence rate, \cite{petrat:2015} where the convergence rate $N^{-1/2}$ was derived for interactions $v(x)=|x|^{-\delta}$ with $0<\delta<3/5$, and \cite{bach:2015} where the convergence rate $N^{-1/6}$ was derived for a class of interactions including $v(x) = |x|^{-1}$. A more detailed overview of the literature is provided at the end of the introduction.

A very interesting open question in a similar direction concerns the extension of the convergence to the fermionic Hartree equations for time scales of order $N^{1/3}$. In other words, by rescaling time, one would like to approximate the dynamics of $\psi^t$ as a solution to the Schr\"odinger equation
\begin{equation}\label{Schr_long_t}
iN^{-1/3} \partial_t \psi^t = \left( -\sum_{i=1}^N \Delta_i + N^{-2/3} \sum_{1\leq i<j \leq N} |x_i-x_j|^{-1} \right) \psi^t,
\end{equation}
for times of order one. On these time scales the forces of order $N^{-1/3}$ have added up to produce a force of order one, such that the dynamics is not free anymore to leading order. Another way of looking at Equation~\eqref{Schr_long_t} follows from rescaling the spatial coordinates $x \to N^{1/3}x$. Then, we are interested in initial conditions $\psi^0$ in a volume of order one, and $\psi^t$ solves
\begin{equation}\label{Schr_sc}
iN^{-1/3} \partial_t \psi^t = \left( -N^{-2/3}\sum_{i=1}^N \Delta_i + N^{-1} \sum_{1\leq i<j \leq N} |x_i-x_j|^{-1} \right) \psi^t.
\end{equation}
In this equation, the small factor $N^{-1/3}$ in front of each derivative is in correspondence to the factor $\hbar$ in the Schr\"odinger equation with SI units. Here, it plays the role of a semiclassical parameter, i.e., for suitable initial conditions, Equation~\eqref{Schr_sc} is a semiclassical equation. Thus, it is natural to assume initial conditions that agree with the structure of the equation. It is an open question to derive the mean-field approximation starting from Equation~\eqref{Schr_sc}. (However, see \cite{porta:2016} for a recent partial result in this direction.) For a class of bounded interactions results were obtained in \cite{erdoes:2004,benedikter:2013,petrat:2015}. Note that for suitable initial conditions the mean-field limit is here coupled to a semiclassical limit. Thus, one would expect that the leading order behavior is described by the classical Vlasov equation. There are several results in this direction, see below.

Let us finish the introduction with an overview of the literature on the subject. The first results for a rigorous derivation of mean-field dynamics were obtained in the 70's by Hepp \cite{hepp:1974} for bosonic systems and by Braun and Hepp \cite{hepp:1977} and Dobrushin \cite{dobrushin:1979} for classical systems. Since then, it has been an active research topic in mathematical physics and a lot of different kinds of rigorous results followed. We refer the interested reader to Spohn's book \cite{spohn:1991} for an overview of classical scaling limits, and, e.g., \cite{benedikter:book} for a collection of results on quantum mechanical scaling limits and for further literature references. The first rigorous results for fermionic systems were obtained by Narnhofer and Sewell \cite{narnhofer:1981}. For a class of regular interactions, they consider both the Schr\"odinger equation \eqref{Schr_lambda} with $\lambda_N=N^{-2/3}$ (``microscopic time-scale'') and \eqref{Schr_sc} (``macroscopic time-scale''). For \eqref{Schr_lambda}, they show locally the closeness to the free evolution, and for \eqref{Schr_sc}, they show closeness to the Vlasov equation in an appropriate sense. The assumptions on the interaction potential were later relaxed by Spohn \cite{spohn:1981}. Many other works deal with the derivation of the Vlasov equation starting from the fermionic Hartree or Hartree-Fock equations \cite{lions:1993,markowich:1993,gasser:1998,pezzotti:2009,athanassoulis:2011,amour:2013,amour:2013_2,benedikter:2015}. A derivation of the Hartree-Fock equations starting from \eqref{Schr_sc} for a class of analytic interaction potentials was given in \cite{erdoes:2004}. There, convergence could be shown only for short times. This result was extended to all times and a larger class of bounded interactions in \cite{benedikter:2013,benedikter:2014}, see also \cite{benedikter:book}, and \cite{petrat:2015} for an alternate proof of some of the results. The result was also extend to cover mixed initial states in \cite{benedikter:2015_2}. Recently, a partial result was proved for \eqref{Schr_sc} (i.e., with Coulomb interaction) \cite{porta:2016}. There, convergence to the Hartree-Fock equations was established, but only for the one special initial wave function with constant $|\varphi_j|$ for all $j$ in a box. Note that the case of Coulomb interaction is technically very involved due to the singularity. There are also results about the derivation of the fermionic Hartree equations starting from \eqref{Schr_lambda} and with coupling constant $\lambda_N = 1/N$. Convergence for bounded interactions was shown in \cite{bardos:2003,bardos:2004,bardos:2007} and for Coulomb interaction in \cite{froehlich:2011}. Results for the limit considered in this article for interactions with a singularity $|x|^{-\delta}$ with $0<\delta<3/5$, and for Coulomb interaction with a singularity cutoff were obtained in \cite{petrat:2015}. In \cite{bach:2015}, a class of interactions including Coulomb interaction and general coupling constants are considered. For the limit we consider in this article, a convergence rate of $N^{-1/6}$ was proved. For larger initial kinetic energies and different coupling constants, it could be shown that the convergence rate is actually better than the one expected from the closeness to a free dynamics with phase factor.

\section{Main Results}

In the main theorems below, we consider the Schr\"odinger equation
\begin{equation}\label{Schr_main}
i \partial_t \psi^t = \left( -\sum_{i=1}^N \Delta_i + N^{-2/3} \sum_{1\leq i<j\leq N} v(x_i-x_j) \right) \psi^t,
\end{equation}
where $v(x) = \pm |x|^{-1}$. As in the introduction, the free equations are
\begin{equation}\label{free_main}
i \partial_t \varphi_j^{\free,t} = -\Delta \varphi_j^{\free,t},
\end{equation}
and the free evolution with phase factor is defined as
\begin{equation}\label{free_with_phase_main}
\phieff_j^t := e^{-i\Phi^t(x)} \varphi_j^{\free,t}, ~~\text{with}~~ \Phi^t(x) = N^{-2/3} \int_0^t ds\, \Big( v*\sum_{j=1}^N|\varphi_j^{\free,s}|^2 \Big)(x).
\end{equation}
Note that it follows that
\begin{equation}\label{eff_eq_main}
i\partial_t \phieff_j^t = \Big(-i\nabla + (\nabla \Phi^t)\Big)^2 \phieff_j^t + N^{-2/3} \big( v * \rho^{\free,t}\big) \phieff_j^t.
\end{equation}

\begin{theorem}\label{thm:main_result1}
Let $\psi^t$ be the solution to the Schr\"odinger equation \eqref{Schr_main} with $v(x)=|x|^{-1}$ or $v(x)=-|x|^{-1}$, with antisymmetric initial condition $\psi^0 \in L^2(\RRR^{3N})$. Let $\phieff_1^t,\ldots,\phieff_N^t$ be solutions to the free equations with phase factor \eqref{free_with_phase_main} with orthonormal initial conditions $\phieff_1^0,\ldots,\phieff_N^0 \in L^2(\RRR^3)$, such that
\begin{equation}\label{ass_Delta4}
\sum_{j=1}^N \norm{\nabla^4 \phieff_j^0}^2 \leq \widetilde{C}N,
\end{equation}
for some $\widetilde{C}>0$. Then there is $C>0$, such that
\begin{equation}\label{main_result1}
\norm[\tr]{\gamma_{\psi^t} - \gamma_{\wedge \phieff_j^t}} \leq Ce^{Ct} \left( \norm[\tr]{\gamma_{\psi^0} - \gamma_{\wedge \phieff_j^0}}^{1/2} + N^{-1/3} \right).
\end{equation}
\end{theorem}

\noindent\textbf{Remarks.}
\begin{enumerate}
\setcounter{enumi}{\theremarks}

\item\label{rem:1} The inequality \eqref{main_result1} still holds if we assume only $\sum_{j=1}^N\norm{\nabla^{3+\varepsilon} \phieff_j^0}^2 \leq \widetilde{C}N$ for some $\varepsilon>0$ instead of \eqref{ass_Delta4}.

\item\label{rem:1_addendum} Two simple examples can be kept in mind for which the assumption \eqref{ass_Delta4} holds. First, one can choose orthonormal $\phieff_1^0,\ldots,\phieff_N^0$ such that for all $j$, $\norm{\nabla^4 \phieff_j^0} \leq C$ for some $C>0$ independent of $j$ and $N$. (The orthonormality can be achieved, e.g., by choosing $\phieff_j^0$'s with non-overlapping compact support.) Second, one can choose plane waves in a box $[0,L]^3$ with $L=N^{1/3}$, i.e., $\phieff_j^0(x) = L^{-3/2} e^{2\pi i L^{-1} k_j \cdot x}$, with $k_j \in \ZZZ^3$ and $k_i\neq k_j$ for all $i \neq j$. Then $\sum_{j=1}^N \norm{\nabla^4 \phieff_j^0}^2 = \sum_{j=1}^N \left(2\pi L^{-1} k_j\right)^8$, i.e., the assumption holds if we choose all (or all but a few) $k_j$'s such that $|k_j| \leq C L$ for some $C>0$.

\item\label{rem:2} Our proof actually yields a slightly better time dependence on the right-hand side of \eqref{main_result1}. For example, one can easily show that the right-hand side is zero at $t=0$, if $\norm[\tr]{\gamma_{\psi^0} - \gamma_{\wedge \phieff_j^0}}=0$.

\item\label{rem:3} The result is also true if we include bounded external fields, e.g., a field that confines the particles to a periodic lattice. Moreover, one can show that it is true for a class of regular confining external potentials that go to infinity for $|x|\to\infty$. Including an external field $W$ means that $-\Delta_i$ is replaced by $-\Delta_i+W(x_i)$ in \eqref{Schr_main}, $-\Delta$ by $-\Delta+W$ in \eqref{free_main}, and that the term $W\phieff_j^t$ is added on the right-hand side of \eqref{eff_eq_main}.

\item\label{rem:4} As discussed after Equation~\eqref{intro_free}, we conjecture that the convergence rate $N^{-1/3}$ can in general not be improved. However, it is hard to prove a lower bound, or to give an explicit example, since there is no easy way to solve the interacting Schr\"odinger equation exactly.

\end{enumerate}
\setcounter{remarks}{\theenumi}

Let us now consider the fermionic Hartree equations
\begin{equation}\label{Hartree_main}
i \partial_t \varphi_j^t = \left( -\Delta + N^{-2/3} \big(v*\rho^t\big) \right) \varphi_j^t,
\end{equation}
where $\rho^t(x) = \sum_{j=1}^N |\varphi_j^t(x)|^2$. Note that the solution theory for this equation is well established, see, e.g., \cite{chadam:1975}.

\begin{theorem}\label{thm:main_result2}
Let $\psi^t$ be the solution to the Schr\"odinger equation \eqref{Schr_main} with $v(x)=|x|^{-1}$ or $v(x)=-|x|^{-1}$, with antisymmetric initial condition $\psi^0 \in L^2(\RRR^{3N})$. Let $\varphi_1^t,\ldots,\varphi_N^t$ be solutions to the fermionic Hartree equations \eqref{Hartree_main} with orthonormal initial conditions $\varphi_1^0,\ldots,\varphi_N^0 \in L^2(\RRR^3)$, such that
\begin{equation}\label{ass_Delta4_result2}
\sum_{j=1}^N \norm{\nabla^4 \varphi_j^0}^2 \leq \widetilde{C}N,
\end{equation}
for some $\widetilde{C}>0$. Then there is $C>0$, such that
\begin{equation}\label{main_result2}
\norm[\tr]{\gamma_{\psi^t} - \gamma_{\wedge \varphi_j^t}} \leq Ce^{e^{Ct}} \left( \norm[\tr]{\gamma_{\psi^0} - \gamma_{\wedge \varphi_j^0}}^{1/2} + N^{-1/2} \right).
\end{equation}
\end{theorem}

\noindent\textbf{Remarks.}
\begin{enumerate}
\setcounter{enumi}{\theremarks}

\item The Remarks~\ref{rem:1}--\ref{rem:3} from below Theorem~\ref{thm:main_result1} also apply here.

\item The theorem also holds if we consider the Hartree-Fock equations instead of \eqref{Hartree_main}, i.e., when there is the additional exchange term $\sum_{k=1}^N \big(v*\overline{\varphi_k^t}\varphi_j^t\big)\varphi_k^t$ on the right-hand side of \eqref{Hartree_main}.

\item The double exponential on the right-hand side of \eqref{main_result2} comes from using two Gronwall estimates. It is physically very undesirable and can possibly be improved using more refined techniques.

\item If we introduce a cutoff around the singularity of the Coulomb interaction, one can improve the rate of convergence from $N^{-1/2}$ to $N^{-1/2-\delta}$ for some $\delta > 0$, depending on the cutoff. This is possible using similar techniques as in the proof of Equation~(2.21) in \cite{benedikter:2013} or as in \cite{mitrouskas:2016_2}. However, we have not been able to improve the convergence rate for Coulomb interaction without cutoff.

\end{enumerate}
\setcounter{remarks}{\theenumi}

\section{Proofs}

In the course of the proofs we use many standard inequalities without references. We refer, e.g., to \cite{liebloss:2001} for proofs.

\absatz

\noindent\textbf{Notation.} In this section, $C$ denotes some positive constant independent of $N$. Since we are usually not interested in precise bounds for constants, $C$ can be different from line to line.

\subsection{Propagation Estimates}

\begin{lemma}\label{lem:v2_free}
Let $\varphi_1^{\free,t},\ldots,\varphi_N^{\free,t}$ be solutions to the free equations \eqref{free_main} with orthonormal initial conditions $\varphi_1^{\free,0},\ldots,\varphi_N^{\free,0} \in L^2(\RRR^3)$ such that
\begin{equation}\label{ass_Delta4_lem}
\sum_{j=1}^N \norm{\nabla^4 \varphi_j^{\free,0}}^2 \leq \widetilde{C} N
\end{equation}
for some $\widetilde{C}>0$. Then there is $C>0$, such that
\begin{equation}\label{v^2_free_bound}
\norm[\infty]{v^2 * \rho^{\free,t}} \leq C N^{1/3}.
\end{equation}
\end{lemma}

\noindent\textbf{Remarks.}
\begin{enumerate}
\setcounter{enumi}{\theremarks}

\item Assuming $\sum_{j=1}^N \norm{\nabla \varphi_j^{\free,0}}^2 \leq \widetilde{C} N$, Hardy's inequality only gives $\norm[\infty]{v^2 * \rho^{\free,t}} \leq C N$, and one can easily construct an example of a $\rho^{\free,t}$ where this bound is sharp in $N$. Here, we would like to prove the best possible dependence in $N$ by using the regularity assumption \eqref{ass_Delta4_lem}. From a physical point of view, this assumption prevents the particles from clustering too much.

\item One can easily check that the $N$ dependence on the right-hand side of \eqref{v^2_free_bound} is optimal. For example, consider a $\rho^{\free,t}$ which is constant in a box of radius $N^{1/3}$. Then equality holds in \eqref{v^2_free_bound}, see also the calculation \eqref{heuristic_nabla_v_star_rho}.

\end{enumerate}
\setcounter{remarks}{\theenumi}

\begin{proof}
We split the area of integration into $B_R = \{ y \in \RRR^3: |y| < R \}$ and the complement $B_R^c=\RRR^3 \backslash B_R$. We then use H\"older's inequality and find
\begin{align}
\big(v^2 * \rho^{\free,t}\big)(x) &= \int_{B_R} v(y)^2 \rho^{\free,t}(x-y) dy + \int_{B_R^c} v(y)^2 \rho^{\free,t}(x-y) dy \nonumber\\
&\leq \left(\int_{B_R} \left(v(y)^2\right)^{11/8} dy \right)^{8/11} \norm[11/3]{\rho^{\free,t}} + \left( \sup_{y \in B_R^c} v(y)^2 \right) \norm[1]{\rho^{\free,t}} \nonumber\\
&\leq (16\pi)^{8/11} R^{2/11} \norm[11/3]{\rho^{\free,t}} + R^{-2} N.
\end{align}
For any $\rho = \sum_{j=1}^N |f_j|^2$ with orthonormal $f_1,\ldots,f_N\in L^2(\RRR^3)$, the generalized Lieb-Thirring inequality \cite{ghidaglia:1988} (see also the references \cite{lieb:1975,lieb:2010,rumin:2011}) reads
\begin{equation}\label{Lieb_Thirring}
\int \rho(x)^{1+2a/3} dx \leq C_a \sum_{j=1}^N \norm{\nabla^{a} f_j}^2.
\end{equation}
Applying this for $a=4$ gives
\begin{equation}
\norm[11/3]{\rho^{\free,t}}^{11/3} = \int \left( \rho^{\free,t}(x) \right)^{11/3} dx \leq C_4 \sum_{j=1}^N \norm{\nabla^4 \varphi_j^{\free,t}}^2.
\end{equation}
Since $i\partial_t\varphi_j^{\free,t} = - \Delta \varphi_j^{\free,t}$, we have that $\norm{\nabla^{a} \varphi_j^{\free,t}}^2 = \norm{\nabla^{a} \varphi_j^{\free,0}}^2$. Together with the assumption $\sum_{j=1}^N \norm{\nabla^4 \varphi_j^{\free,0}}^2 \leq \widetilde{C} N$ we thus find
\begin{equation}
\big(v^2 * \rho^{\free,t}\big)(x) \leq C R^{2/11} N^{3/11} + R^{-2} N.
\end{equation}
Optimizing $R$ leads to $R \propto N^{1/3}$ and gives the desired bound \eqref{v^2_free_bound}.
\end{proof}

\begin{lemma}\label{lem:v2_hartree}
Let $\varphi_1^t,\ldots,\varphi_N^t$ be solutions to the fermionic Hartree equations \eqref{Hartree_main} with orthonormal initial conditions $\varphi_1^0,\ldots,\varphi_N^0 \in L^2(\RRR^3)$ such that
\begin{equation}\label{ass_Delta4_lem2}
\sum_{j=1}^N \norm{\nabla^4 \varphi_j^0}^2 \leq \widetilde{C} N
\end{equation}
for some $\widetilde{C}>0$. Then there is $C>0$, such that
\begin{equation}
\norm[\infty]{v^2 * \rho^t} \leq C e^{Ct} N^{1/3}.
\end{equation}
\end{lemma}

\begin{proof}
For any two orthonormal families $\varphi_1,\ldots,\varphi_N \in L^2(\RRR^3)$ and  $\phieff_1,\ldots,\phieff_N \in L^2(\RRR^3)$, we define the corresponding densities $\rho(x)=\sum_{j=1}^N|\varphi_j(x)|^2$ and $\rhoeff(x)=\sum_{j=1}^N|\phieff_j(x)|^2$. In the following, we use the notation $\lilscp{\varphi}{f_{x\cdot}\chi} = (f*\overline{\varphi}\chi)(x)$. Using Cauchy-Schwarz and in the last step Hardy's inequality, we find
\begin{align}\label{v2rho_splitting}
\big(v^2 * \rho\big)(x) &=  \big(v^2 * \big(\rho - \rhoeff\big)\big)(x) + \big(v^2 * \rhoeff\big)(x) \nonumber\\
&= \sum_{j=1}^N \bigg( \scp{\varphi_j}{v_{x\cdot}^2\varphi_j} - \scp{\phieff_j}{v_{x\cdot}^2\phieff_j} \bigg) + \big(v^2 * \rhoeff\big)(x) \nonumber\\
&= \sum_{j=1}^N \bigg( \scp{(\varphi_j - \phieff_j)}{v_{x\cdot}^2(\varphi_j - \phieff_j)} + \scp{(\varphi_j - \phieff_j)}{v_{x\cdot}^2\phieff_j} + \scp{\phieff_j}{v_{x\cdot}^2(\varphi_j-\phieff_j)} \bigg) + \big(v^2 * \rhoeff\big)(x) \nonumber\\
&\leq 2 \sum_{j=1}^N \Big( v^2 * |\varphi_j-\phieff_j|^2 \Big)(x) + 2\big(v^2 * \rhoeff\big)(x) \nonumber \\
&\leq 8\sum_{j=1}^N \norm{\nabla (\varphi_j-\phieff_j)}^2 + 2 \big(v^2 * \rhoeff\big)(x).
\end{align}
We now apply \eqref{v2rho_splitting} for $\varphi_j=\varphi_j^t$, i.e., the solution to the fermionic Hartree equations \eqref{Hartree_main}, and $\phieff_j=\phieff_j^t$, i.e., the solution to the free equation with phase factor \eqref{free_with_phase_main}, and with the same initial conditions $\varphi_j^0 = \phieff_j^0$. Since $|\phieff_j^t| = |\varphi_j^{\free,t}|$ (when the initial conditions are the same), we can apply Lemma~\ref{lem:v2_free}, i.e.,
\begin{equation}
\big(v^2 * \rhoeff\big)(x) \leq CN^{1/3}.
\end{equation}
We now control the remaining term with a Gronwall estimate. We proceed in three steps. First, we collect some inequalities necessary to control the terms that appear in the time derivative. Then, in the second step, we control $\sum_{j=1}^N \norm{\varphi_j^t-\phieff_j^t}^2$. In the third step, we estimate $\partial_t \sum_{j=1}^N \norm{\nabla (\varphi_j^t-\phieff_j^t)}^2$ to obtain the desired bound.

\absatz

\noindent \textbf{Step 1:} For orthonormal families $\{ \varphi^i_j \}_{j=1,\ldots,N}$, we define the corresponding densities $\rho_i(x)=\sum_{j=1}^N |\varphi^i_j(x)|^2$ and the quantities $\rho^{\nabla}_i(x)=\sum_{j=1}^N |\nabla \varphi^i_j(x)|^2$, $i=1,2,3$. First, note that by Cauchy-Schwarz
\begin{align}
\nabla \rho_1 = \sum_{j=1}^N \left( (\overline{\nabla\varphi^1_j}) \varphi^1_j + \overline{\varphi^1_j} \nabla\varphi^1_j \right) \leq 2 \sum_{j=1}^N |\nabla\varphi^1_j| |\varphi^1_j| \leq 2 \sqrt{\rho_1} \sqrt{\rho_1^{\nabla}}.
\end{align}
It follows that
\begin{align}\label{nrho_nrho_rho}
\norm[1]{(\nabla\rho_1)(\nabla\rho_2)\rho_3} \leq 4 \norm[1]{\sqrt{\rho_1^{\nabla}} \sqrt{\rho_1} \sqrt{\rho_2^{\nabla}} \sqrt{\rho_2} \sqrt{\rho_3} \sqrt{\rho_3}} \leq 4 \sqrt{\norm[1]{\rho_1^{\nabla} \rho_2 \rho_3}} \sqrt{\norm[1]{\rho_1 \rho_2^{\nabla} \rho_3}}.
\end{align}
Using first H\"older's inequality and then Sobolev's inequality and the Lieb-Thirring inequality \eqref{Lieb_Thirring} for $a=3$, we find
\begin{align}\label{nrho_rho_rho}
\norm[1]{\rho_1^{\nabla}\rho_2\rho_3} &\leq \sum_{j=1}^N \left( \int |\varphi_j^1|^6 \right)^{1/3} \left( \int \rho_2^3  \right)^{1/3} \left( \int \rho_3^3  \right)^{1/3} \nonumber\\
&\leq C \sum_{j=1}^N \norm{\Delta \varphi_j^1}^2 \left( \sum_{j=1}^N \norm{\nabla^3 \varphi_j^2}^2 \right)^{1/3} \left( \sum_{j=1}^N \norm{\nabla^3 \varphi_j^3}^2 \right)^{1/3}.
\end{align}
Using only the Lieb-Thirring inequality \eqref{Lieb_Thirring} for $a=3$, we find analogously
\begin{align}\label{rho123}
\norm[1]{\rho_1\rho_2\rho_3} \leq C \left( \sum_{j=1}^N \norm{\nabla^3 \varphi_j^1}^2 \right)^{1/3} \left( \sum_{j=1}^N \norm{\nabla^3 \varphi_j^2}^2 \right)^{1/3} \left( \sum_{j=1}^N \norm{\nabla^3 \varphi_j^3}^2 \right)^{1/3}.
\end{align}
We also need to control gradients of the phase factor $\Phi^t(x) = N^{-2/3} \int_0^t ds\, \Big( v*\sum_{j=1}^N|\varphi_j^{\free,s}|^2 \Big)(x)$. With Lemma~\ref{lem:v2_free} we find
\begin{align}\label{nablaPhi}
\norm[\infty]{\nabla \Phi^t} = N^{-2/3} \sup_x \int_0^t ds \Big( (\nabla v)*\rho^{\free,s} \Big)(x) \leq \sqrt{3}\, t N^{-2/3} \sup_{s\in [0,t]} \norm[\infty]{v^2 * \rho^{\free,s}} \leq Ct N^{-1/3}.
\end{align}
The terms $\sum_{j=1}^N \norm{\nabla \phieff_j^t}^2$ and $\sum_{j=1}^N \norm{\Delta \phieff_j^t}^2$ can easily be controlled with the assumption \eqref{ass_Delta4_lem2} and the preceding estimates. We find
\begin{align}\label{nabla_phieff}
\sum_{j=1}^N \norm{\nabla \phieff_j^t}^2 &= \sum_{j=1}^N \norm{(-i\nabla \Phi^t) e^{-i\Phi^t} \varphi_j^{\free,t} + e^{-i\Phi^t} \nabla \varphi_j^{\free,t}}^2 \nonumber\\
&\leq 2 N \norm[\infty]{\nabla \Phi^t} + 2 \sum_{j=1}^N \norm{\nabla \varphi_j^{\free,t}}^2 \nonumber\\
&\leq C (1+t) N.
\end{align}
Note that due to $\Delta_x |x-y|^{-1} = -4\pi\delta(x-y)$ we have
\begin{equation}\label{DeltaPhi}
\Delta_x \Phi^t(x) = -4\pi N^{-2/3} \int_0^t ds \, \rho^{\free,s}(x).
\end{equation}
Using the assumption \eqref{ass_Delta4_lem2}, and \eqref{rho123}, \eqref{nablaPhi} and \eqref{DeltaPhi}, it follows that
\begin{align}\label{Laplace_phieff}
\sum_{j=1}^N \norm{\Delta \phieff_j^t}^2 &= \sum_{j=1}^N \norm{-i(\Delta \Phi^t) e^{-i\Phi^t} \varphi_j^{\free,t} - (\nabla \Phi^t)^2 e^{-i\Phi^t} \varphi_j^{\free,t} -2i (\nabla \Phi^t) e^{-i\Phi^t} \nabla\varphi_j^{\free,t} + e^{-i\Phi^t} \Delta \varphi_j^{\free,t}}^2 \nonumber\\
&\leq 4 \left( \int \rho^{\free,t}(\Delta \Phi^t)^2 + N \norm[\infty]{\nabla \Phi^t}^4 + 2\norm[\infty]{\nabla \Phi^t}^2 \sum_{j=1}^N \norm{\nabla \varphi_j^{\free,t}}^2 + \sum_{j=1}^N \norm{\Delta \varphi_j^{\free,t}}^2 \right) \nonumber\\
&\leq C \left( N^{-4/3} N t^2 + t^4 N^{-4/3} N + t^2 N^{-2/3} N + N \right) \nonumber\\
&\leq C(1+t^4)N.
\end{align}
Finally, we need to control the kinetic energy of the solution to the Hartree equation \eqref{Hartree_main}, which can be done using energy conservation, i.e., $E^t=E^0$, where
\begin{equation}
E^t := \sum_{j=1}^N \norm{\nabla \varphi_j^t}^2 + \frac{1}{2} N^{-2/3} \int (v*\rho^t)\rho^t.
\end{equation}
By the Hardy-Littlewood-Sobolev, H\"older and Lieb-Thirring inequalities, we have
\begin{align}\label{HLS_v}
\int (v*\rho)\rho \leq C \norm[6/5]{\rho}^2 \leq C \left( \int \rho^{5/3} \right)^{1/2} \left( \int \rho \right)^{7/6} \leq C \left( \sum_{j=1}^N \norm{\nabla \varphi_j}^2 \right)^{1/2} N^{7/6},
\end{align}
and thus
\begin{align}
\sum_{j=1}^N \norm{\nabla \varphi_j^t}^2 &= E^t - \frac{1}{2} N^{-2/3} \int (v*\rho^t)\rho^t \nonumber\\
&= E^0 - \frac{1}{2} N^{-2/3} \int (v*\rho^t)\rho^t \nonumber\\
&= \sum_{j=1}^N \norm{\nabla \varphi_j^0}^2 + \frac{1}{2} N^{-2/3} \int (v*\rho^0)\rho^0 - \frac{1}{2} N^{-2/3} \int (v*\rho^t)\rho^t \nonumber\\
&\leq \sum_{j=1}^N \norm{\nabla \varphi_j^0}^2 + \frac{C}{2} N^{-2/3}\left( \sum_{j=1}^N \norm{\nabla \varphi_j^0}^2 \right)^{1/2} N^{7/6} + \frac{C}{2} N^{-2/3}\left( \sum_{j=1}^N \norm{\nabla \varphi_j^t}^2 \right)^{1/2} N^{7/6} \nonumber\\
&\leq \frac{3}{2} \sum_{j=1}^N \norm{\nabla \varphi_j^0}^2 + \frac{1}{2} \sum_{j=1}^N \norm{\nabla \varphi_j^t}^2 + CN,
\end{align}
i.e., 
\begin{align}\label{conservation_hartree_kin}
\sum_{j=1}^N \norm{\nabla \varphi_j^t}^2 &\leq 3 \sum_{j=1}^N \norm{\nabla \varphi_j^0}^2 + CN.
\end{align}
Similarly to \eqref{HLS_v}, by the Hardy-Littlewood-Sobolev, H\"older and Lieb-Thirring inequalities, we have
\begin{align}\label{HLS_v2}
\int (v^2*\rho)\rho \leq C \norm[3/2]{\rho}^2 \leq C \left( \int \rho^{5/3} \right) \left( \int \rho \right)^{1/3} \leq C \left( \sum_{j=1}^N \norm{\nabla \varphi_j}^2 \right) N^{1/3}.
\end{align}

\noindent \textbf{Step 2:} We now control $\sum_{j=1}^N \norm{\varphi_j^t-\phieff_j^t}^2$ with a Gronwall estimate. We find
\begin{align}
\partial_t \sum_{j=1}^N \norm{\varphi_j^t-\phieff_j^t}^2 =& 2N^{-2/3}\Im \sum_{j=1}^N \scp{(\varphi_j^t - \phieff_j^t)}{\Big[ \big(v*\rho^t\big)\varphi_j^t - \big(v*\rho^{\free,t}\big)\phieff_j^t \Big]} \nonumber\\
& + 2\Im \sum_{j=1}^N \scp{(\varphi_j^t - \phieff_j^t)}{\Big[ i (\Delta\Phi^t)\phieff_j^t + 2i (\nabla\Phi^t)\nabla\phieff_j^t - (\nabla\Phi^t)^2\phieff_j^t \Big]}.
\end{align}
The first term can be controlled due to cancellations of the mean-field terms. Using first that we take only the imaginary part, and then Cauchy-Schwarz and Lemma~\ref{lem:v2_free}, we find
\begin{align}
& 2N^{-2/3}\Im \sum_{j=1}^N \scp{(\varphi_j^t - \phieff_j^t)}{\Big[ \big(v*\rho^t\big)\varphi_j^t - \big(v*\rho^{\free,t}\big)\phieff_j^t \Big]} \nonumber\\
&\qquad= 2N^{-2/3}\Im \sum_{j=1}^N \scp{(\varphi_j^t - \phieff_j^t)}{\Big(v*\big(\rho^t-\rho^{\free,t}\big)\Big)\phieff_j^t} \nonumber\\
&\qquad= 2N^{-2/3}\Im \sum_{j,k=1}^N \scp{(\varphi_j^t - \phieff_j^t)}{\Big(v*\big((\overline{\varphi_k^t-\phieff_k^t})\varphi_k\big)\Big)\phieff_j^t} \nonumber\\
&\qquad\leq 2N^{-2/3} \sum_{j=1}^N \norm{\varphi_j^t-\phieff_j^t}^2 \sqrt{N \norm[\infty]{v^2 * \rho^{\free,t}}} \nonumber\\
&\qquad\leq C \sum_{j=1}^N \norm{\varphi_j^t-\phieff_j^t}^2.
\end{align}
The second term can be controlled due to the smallness of gradients of $\Phi^t$. Using Cauchy-Schwarz and then \eqref{rho123}, \eqref{nablaPhi}, \eqref{nabla_phieff} and \eqref{DeltaPhi}, we find
\begin{align}
& 2\Im \sum_{j=1}^N \scp{(\varphi_j^t - \phieff_j^t)}{\Big[ i (\Delta\Phi^t)\phieff_j^t + 2i (\nabla\Phi^t)\nabla\phieff_j^t - (\nabla\Phi^t)^2\phieff_j^t \Big]} \nonumber\\
&\qquad\leq 2\sqrt{\sum_{j=1}^N \norm{\varphi_j^t-\phieff_j^t}^2} \sqrt{ 3\sum_{j=1}^N \bigg( \norm{(\Delta\Phi^t)\phieff_j^t}^2 + 2 \norm{(\nabla\Phi^t)\nabla\phieff_j^t}^2 + \norm{(\nabla\Phi^t)^2 \phieff_j^t}^2 \bigg)} \nonumber\\
&\qquad\leq \sum_{j=1}^N \norm{\varphi_j^t-\phieff_j^t}^2 + C t^2 N^{-4/3} N + Ct^2N^{-2/3} (1+t) N + Ct^4N^{-4/3} N \nonumber\\
&\qquad\leq \sum_{j=1}^N \norm{\varphi_j^t-\phieff_j^t}^2 + C (1+t^4) N^{1/3}.
\end{align}
In total, we have estimated
\begin{align}\label{dtphi_phieff}
\partial_t \sum_{j=1}^N \norm{\varphi_j^t-\phieff_j^t}^2 \leq C \sum_{j=1}^N \norm{\varphi_j^t-\phieff_j^t}^2 + C(1+t^4)N^{1/3}.
\end{align}
Gronwall's inequality states that if for a function $f:\RRR\to\RRR$ the estimate
\begin{equation}
\partial_t f(t) \leq C(t) f(t) + \varepsilon(t)
\end{equation}
holds for some $C:\RRR\to\RRR^+$ and some $\varepsilon:\RRR\to\RRR^+$, then
\begin{equation}\label{gronwall}
f(t) \leq e^{\int_0^tC(s)ds} f(0) + \int_0^t \varepsilon(s) e^{\int_s^tC(s')ds'}ds.
\end{equation}
Thus, since $\norm{\varphi_j^0-\phieff_j^0} = 0$, the estimate \eqref{dtphi_phieff} implies that
\begin{align}\label{estimate_phi_phieff}
\sum_{j=1}^N \norm{\varphi_j^t-\phieff_j^t}^2 \leq C N^{1/3} e^{Ct} \int_0^t (1+s^4) e^{-Cs} ds \leq C N^{1/3} e^{Ct}.
\end{align}

\noindent \textbf{Step 3:} We can now control $\sum_{j=1}^N \norm{\nabla(\varphi_j^t-\phieff_j^t)}^2$ with a Gronwall estimate. We find
\begin{align}\label{dt_nabla_phi_phieff}
\partial_t \sum_{j=1}^N \norm{\nabla (\varphi_j^t-\phieff_j^t)}^2 =& 2N^{-2/3}\Im \sum_{j=1}^N \scp{\nabla (\varphi_j^t-\phieff_j^t)}{\nabla \Big[ \big(v*\rho^t\big)\varphi_j^t - \big(v*\rho^{\free,t}\big)\phieff_j^t \Big]} \nonumber\\
& + 2\Im \sum_{j=1}^N \scp{\nabla (\varphi_j^t-\phieff_j^t)}{\nabla \Big[ i (\Delta\Phi^t)\phieff_j^t + 2i (\nabla\Phi^t)\nabla\phieff_j^t - (\nabla\Phi^t)^2\phieff_j^t \Big]}.
\end{align}
Similarly to before, the first term can be controlled by using cancellations of the mean-field terms. Since we only take the imaginary part, we can write
\begin{align}\label{dt_nabla_first_term}
& 2N^{-2/3}\Im \sum_{j=1}^N \scp{\nabla (\varphi_j^t-\phieff_j^t)}{\nabla \Big[ \big(v*\rho^t\big)\varphi_j^t - \big(v*\rho^{\free,t}\big)\phieff_j^t \Big]} \nonumber\\
&\qquad= 2N^{-2/3}\Im \sum_{j=1}^N \scp{\nabla (\varphi_j^t-\phieff_j^t)}{\Big[ \big((\nabla v)*\rho^t\big)\varphi_j^t - \big((\nabla v)*\rho^{\free,t}\big)\phieff_j^t \Big]} \nonumber\\
&\qquad\quad + 2N^{-2/3}\Im \sum_{j=1}^N \scp{\nabla (\varphi_j^t-\phieff_j^t)}{\Big[ \big(v*\rho^t\big)\nabla\phieff_j^t - \big(v*\rho^{\free,t}\big)\nabla\phieff_j^t \Big]} \nonumber\\
&\qquad= 2N^{-2/3}\Im \sum_{j=1}^N \scp{\nabla (\varphi_j^t-\phieff_j^t)}{\Big[ \big((\nabla v)*\rho^t\big)\varphi_j^t - \big((\nabla v)*\rho^{\free,t}\big)\phieff_j^t \Big]} \nonumber\\
&\qquad\quad + 2N^{-2/3}\Im \sum_{j,k=1}^N \scp{\nabla (\varphi_j^t-\phieff_j^t)}{\bigg(v*\Big( |\varphi_k^t-\phieff_k^t|^2 + (\overline{\varphi_k^t-\phieff_k^t})\phieff_k^t + \overline{\phieff^t_k}(\varphi_k^t-\phieff_k^t) \Big) \bigg)\nabla\phieff_j^t}.
\end{align}
By Cauchy-Schwarz, \eqref{v2rho_splitting}, Lemma~\ref{lem:v2_free}, and \eqref{HLS_v2} together with \eqref{conservation_hartree_kin} and the assumption \eqref{ass_Delta4_lem2}, we find
\begin{align}
& 2N^{-2/3}\Im \sum_{j=1}^N \scp{\nabla (\varphi_j^t-\phieff_j^t)}{\big((\nabla v)*\rho^t\big)\varphi_j^t} \nonumber\\
&\qquad\leq 2N^{-2/3} \sqrt{\sum_{j=1}^N \norm{\nabla (\varphi_j^t-\phieff_j^t)}^2 \norm[\infty]{v^2*\rho^t}} \sqrt{\int (v^2*\rho^t)\rho^t} \nonumber\\
&\qquad\leq C N^{-2/3} \sqrt{\sum_{j=1}^N \norm{\nabla (\varphi_j^t-\phieff_j^t)}^2} \sqrt{\sum_{j=1}^N \norm{\nabla (\varphi_j^t-\phieff_j^t)}^2 + \norm[\infty]{v^2*\rho^{\free,t}}} \sqrt{N^{4/3}} \nonumber\\
&\qquad\leq C \sum_{j=1}^N \norm{\nabla (\varphi_j^t-\phieff_j^t)}^2 + CN^{1/3}.
\end{align}
In a similar way, by using only Lemma~\ref{lem:v2_free} after applying Cauchy-Schwarz, we find
\begin{align}
2N^{-2/3}\Im \sum_{j=1}^N \scp{\nabla (\varphi_j^t-\phieff_j^t)}{\big((\nabla v)*\rho^{\free,t}\big)\phieff_j^t} \leq C \sum_{j=1}^N \norm{\nabla (\varphi_j^t-\phieff_j^t)}^2 + CN^{1/3}.
\end{align}
Furthermore, using Cauchy-Schwarz and Hardy's inequality in the first step, and then Lemma~\ref{lem:v2_free}, \eqref{nabla_phieff}, \eqref{Laplace_phieff} and the result \eqref{estimate_phi_phieff} from step 2, we find
\begin{align}
& 2N^{-2/3}\Im \sum_{j,k=1}^N \scp{\nabla (\varphi_j^t-\phieff_j^t)}{\bigg(v*\Big( |\varphi_k^t-\phieff_k^t|^2 + (\overline{\varphi_k^t-\phieff_k^t})\phieff_k^t + \overline{\phieff^t_k}(\varphi_k^t-\phieff_k^t) \Big) \bigg)\nabla\phieff_j^t} \nonumber\\
&\qquad\leq 2N^{-2/3} \sqrt{\sum_{j=1}^N \norm{\nabla (\varphi_j^t-\phieff_j^t)}^2 \sum_{k=1}^N \norm{\varphi_k^t-\phieff_k^t}^2} \sqrt{\sum_{j=1}^N \norm{\Delta \phieff_j^t}^2\sum_{k=1}^N \norm{\varphi_k^t-\phieff_k^t}^2} \nonumber\\
&\qquad\quad + 2N^{-2/3} \sqrt{\sum_{j=1}^N \norm{\nabla (\varphi_j^t-\phieff_j^t)}^2 \sum_{k=1}^N \norm{\varphi_k^t-\phieff_k^t}^2} \sqrt{\sum_{j=1}^N \norm{\nabla \phieff_j^t}^2 \norm[\infty]{v^2*\rho^{\free,t}}} \nonumber\\
&\qquad\quad + 2N^{-2/3} \sqrt{\sum_{j=1}^N \norm{\nabla (\varphi_j^t-\phieff_j^t)}^2 \norm[\infty]{v^2*\rho^{\free,t}}} \sqrt{\sum_{j=1}^N \norm{\nabla \phieff_j^t}^2\sum_{k=1}^N \norm{\varphi_k^t-\phieff_k^t}^2} \nonumber\\
&\qquad\leq C \sum_{j=1}^N \norm{\nabla (\varphi_j^t-\phieff_j^t)}^2 + C N^{-4/3} \left( N^{1/3} e^{Ct} \right)^2 (1+t^4)N  + C N^{-4/3} N^{1/3} e^{Ct} (1+t)N N^{1/3} \nonumber\\
&\qquad\quad + C N^{-4/3} N^{1/3} (1+t)N N^{1/3} e^{Ct} \nonumber\\
&\qquad\leq C \sum_{j=1}^N \norm{\nabla (\varphi_j^t-\phieff_j^t)}^2 + C N^{1/3} (1+t^4) e^{Ct}.
\end{align}
This completes the estimate for the first term on the right-hand side of \eqref{dt_nabla_phi_phieff}. For the second term on the right-hand side of \eqref{dt_nabla_phi_phieff}, we find, using Cauchy-Schwarz,
\begin{align}
& 2\Im \sum_{j=1}^N \scp{\nabla (\varphi_j^t-\phieff_j^t)}{\nabla \Big[ i (\Delta\Phi^t)\phieff_j^t + 2i (\nabla\Phi^t)\nabla\phieff_j^t - (\nabla\Phi^t)^2\phieff_j^t \Big]} \nonumber\\
&\qquad= 2\Im \sum_{j=1}^N \scp{\nabla (\varphi_j^t-\phieff_j^t)}{\bigg[ i (\nabla\Delta\Phi^t)\phieff_j^t + 3i (\Delta\Phi^t)\nabla\phieff_j^t + 2i (\nabla\Phi^t)\Delta\phieff_j^t - 2(\Delta\Phi^t)(\nabla\Phi^t)\phieff_j^t - (\nabla\Phi^t)^2\nabla\phieff_j^t \bigg]} \nonumber\\
&\qquad\leq \sum_{j=1}^N \norm{\nabla (\varphi_j^t-\phieff_j^t)}^2 + C \int(\nabla\Delta\Phi^t)^2 \rho^{\free,t} + C \int (\Delta\Phi^t)^2 \sum_{j=1}^N|\nabla\phieff_j^t|^2 + C \norm[\infty]{\nabla\Phi^t}^2 \sum_{j=1}^N \norm{\Delta \phieff_j}^2 \nonumber\\
&\qquad\quad + C \norm[\infty]{\nabla\Phi^t}^2 \int (\Delta\Phi^t)^2 \rho^{\free,t} + C \norm[\infty]{\nabla\Phi^t}^4 \sum_{j=1}^N \norm{\nabla \phieff_j}^2.
\end{align}
The remaining terms can be estimated using the inequalities from step 1. Using \eqref{DeltaPhi}, and then \eqref{nrho_nrho_rho} with \eqref{nrho_rho_rho} and the assumption \eqref{ass_Delta4_lem2}, we find
\begin{align}
\int(\nabla\Delta\Phi^t)^2 \rho^{\free,t} \leq Ct^2 N^{-4/3}N^{5/3} = Ct^2 N^{1/3}.
\end{align}
Using \eqref{DeltaPhi}, and then \eqref{nrho_rho_rho} with the assumption \eqref{ass_Delta4_lem2} and with \eqref{Laplace_phieff}, we find
\begin{align}
\int (\Delta\Phi^t)^2 \sum_{j=1}^N|\nabla\phieff_j^t|^2 \leq C t^2 N^{-4/3} N^{1/3} N^{1/3} (1+t^4)N \leq C (1+t^6) N^{1/3}.
\end{align}
By \eqref{nablaPhi} and \eqref{Laplace_phieff}, we have
\begin{align}
\norm[\infty]{\nabla\Phi^t}^2 \sum_{j=1}^N \norm{\Delta \phieff_j}^2 \leq C \left(tN^{-1/3}\right)^2 (1+t^4) N \leq C (1+t^6) N^{1/3}.
\end{align}
By \eqref{nablaPhi} and \eqref{rho123} with the assumption \eqref{ass_Delta4_lem2}, we have
\begin{align}\label{Delta_Phi_2_rho_free}
\norm[\infty]{\nabla\Phi^t}^2 \int (\Delta\Phi^t)^2 \rho^{\free,t} \leq C \left( tN^{-1/3} \right)^2 N^{-4/3}t^2N = Ct^4 N^{-1}.
\end{align}
Finally, using \eqref{nablaPhi} and \eqref{nabla_phieff}, we find
\begin{align}
\norm[\infty]{\nabla\Phi^t}^4 \sum_{j=1}^N \norm{\nabla \phieff_j}^2 \leq C \left( tN^{-1/3} \right)^4 (1+t)N \leq C (1+t^5) N^{-1/3}.
\end{align}
Collecting all of our estimates, we have found
\begin{align}
\partial_t \sum_{j=1}^N \norm{\nabla (\varphi_j^t-\phieff_j^t)}^2 \leq C \sum_{j=1}^N \norm{\nabla (\varphi_j^t-\phieff_j^t)}^2 + C N^{1/3} (1+t^6)e^{Ct}.
\end{align}
Since $\norm{\nabla (\varphi_j^0-\phieff_j^0)} = 0$, Gronwall's inequality \eqref{gronwall} gives
\begin{align}
\sum_{j=1}^N \norm{\nabla (\varphi_j^t-\phieff_j^t)}^2 \leq C N^{1/3} e^{Ct}\int_0^t (1+s^6)e^{Cs} e^{-Cs} ds \leq C N^{1/3} e^{Ct}.
\end{align}
Looking back at \eqref{v2rho_splitting}, we have proved the lemma.
\end{proof}

\subsection{Proof of Main Theorems}

Before we come to the proofs of the main theorems, let us briefly summarize the main elements of Pickl's $\alpha$-method \cite{pickl:2011method} for fermions in the most simple form, see \cite{petrat:2015,bach:2015} for more details. First, we define the projectors
\begin{equation}
p := \sum_{j=1}^N\ketbr{\varphi_j} = N\gamma_{\wedge \varphi_j}, \qquad \text{and}\qquad q = 1-p.
\end{equation}
For a one-particle operator $A:L^2(\RRR^3)\to L^2(\RRR^3)$, we use the notation $A_1$ to indicate that the operator acts only on the first variable, i.e., we define 
\begin{equation}
A_1:L^2(\RRR^{3N})\to L^2(\RRR^{3N}), A_1 = A\otimes 1 \otimes \ldots\otimes 1.
\end{equation}
Instead of controlling the difference of the reduced density matrices in trace norm directly, it is technically better to control the functional
\begin{equation}
\alpha(\psi,\{ \varphi_j \}_{j=1,\ldots,N}) := \scp{\psi}{q_1 \psi},
\end{equation}
which was already introduced (in a static setting) in \cite{bach:1993,graf_solovej:1994}. This functional is related to the difference in reduced density matrices by
\begin{align}\label{alpha_gamma}
\norm[\tr]{\gamma_{\psi} - \gamma_{\wedge \varphi_j}}^2 \leq 8 \alpha(\psi,\{ \varphi_j \}_{j=1,\ldots,N}) \leq 4 \norm[\tr]{\gamma_{\psi} - \gamma_{\wedge \varphi_j}},
\end{align}
see \cite[Lemma~3.2]{petrat:2015}.

\begin{proof}[Proof of Theorem~\ref{thm:main_result1}]
In the proof we would like to control $\alphaeff(t) := \alpha(\psi^t,\{ \phieff_j^t \}_{j=1,\ldots,N})$, where $\psi^t$ is solution to the Schr\"odinger equation \eqref{Schr_main}, and $\phieff_j^t$ are the solutions to \eqref{free_with_phase_main}, i.e., 
\begin{equation}
i\partial_t \phieff_j^t = \Big(-i\nabla + (\nabla \Phi^t)\Big)^2 \phieff_j^t + N^{-2/3} \big( v * \rho^{\free,t}\big) \phieff_j^t =: h^{\eff,t}\phieff_j^t.
\end{equation}
Accordingly, we define $\peff^t = \sum_{j=1}^N\ketbr{\phieff_j^t}$ and $\qeff^t = 1 - \peff^t$. Using the antisymmetry of $\psi^t$, we find for the time derivative of $\alphaeff(t)$,
\begin{align}
\partial_t \alphaeff(t) &= i \scp{\psi^t}{\Big[ H - h_1^{\eff,t} , \qeff_1^{t} \Big] \psi^t} \nonumber \\
&= \underbrace{i N^{-2/3}\scp{\psi^t}{\bigg[ \frac{N-1}{2} v_{12} - \Big( v*\rho^{\free,t} \Big)_1 , \qeff_1^{t} \bigg] \psi^t}}_{:= R_1(t)} \nonumber \\
&\qquad + \underbrace{i \scp{\psi^t}{\Big[- \Delta_1 - \big( -i\nabla_1 + (\nabla\Phi^t)_1 \big)^2 , \qeff_1^{t} \Big] \psi^t}}_{:= R_2(t)}.
\end{align}

\noindent\textbf{Control of $R_1(t)$.} Since $\rho^{\free,t} = \rhoeff^t$, we see that the term $R_1(t)$ has already been controlled in \cite[Lemma~7.4 for $\gamma=1$]{petrat:2015}. There, the bound
\begin{equation}\label{bound_R_1_1}
R_1(t) \leq C N^{-1/6} \sqrt{ \norm[\infty]{v^2 * \rhoeff^t}} \, \Big( \alphaeff(t) + N^{-1} \Big)
\end{equation}
was established. Since now again $\rhoeff^t = \rho^{\free,t}$, we can use Lemma~\ref{lem:v2_free} to conclude
\begin{equation}\label{bound_R_1_2}
R_1(t) \leq C \Big( \alphaeff(t) + N^{-1} \Big).
\end{equation}

\noindent\textbf{Control of $R_2(t)$.} In order to control $R_2(t)$ we use the smallness of gradients of the phase factor $\Phi^t$. Note that due to the antisymmetry of $\psi^t$, we have $\norm{A_1\psi^t}^2 \leq N^{-1} \tr(A^*A)$, see, e.g., \cite[Lemma~3.12]{bach:2015}. We find
\begin{align}\label{bound_R_2}
R_2(t) &= 2\Im \scp{\psi^t}{\qeff_1^{t}\Big(- \Delta_1 - \big( -i\nabla_1 + (\nabla\Phi^t)_1 \big)^2 \Big) \peff_1^{t} \psi^t} \nonumber \\
&= 2\Im \scp{\psi^t}{\qeff_1^{t}\bigg( i(\Delta\Phi^t)_1 + 2 i (\nabla\Phi^t)_1 \nabla_1 - (\nabla\Phi^t)_1^2 \bigg) \peff_1^{t} \psi^t} \nonumber \\
&\leq 2\norm{\qeff_1^{t}\psi^t} \left( \sqrt{\scp{\psi^t}{\peff_1^{t} (\Delta\Phi^t)_1^2 \peff_1^{t} \psi^t}} + \norm[\infty]{\nabla\Phi^t} \sqrt{\scp{\psi^t}{\peff_1^{t}(-\Delta_1)\peff_1^{t}\psi^t}} + \norm[\infty]{\nabla\Phi^t}^2 \sqrt{\scp{\psi^t}{\peff_1^{t}\psi^t}} \right) \nonumber\\
&\leq \alphaeff(t) + C N^{-1} \int (\Delta\Phi^t)^2\rho^{\free,t} + C \norm[\infty]{\nabla\Phi^t}^2 N^{-1} \sum_{j=1}^N \norm{\nabla\phieff_j^t}^2 + C \norm[\infty]{\nabla\Phi^t}^4.
\end{align}
In order to bound these terms, we refer back to the proof of Lemma~\ref{lem:v2_hartree}. In the same way as in the estimate \eqref{Delta_Phi_2_rho_free}, we find that
\begin{equation}
N^{-1} \int (\Delta\Phi^t)^2\rho^{\free,t} \leq N^{-1} N^{-4/3} Ct^2 N = Ct^2 N^{-4/3}.
\end{equation}
From \eqref{nablaPhi} and \eqref{nabla_phieff}, we find
\begin{equation}
\norm[\infty]{\nabla\Phi^t}^2 N^{-1} \sum_{j=1}^N \norm{\nabla\phieff_j^t}^2 \leq C \left( tN^{-1/3} \right)^2 N^{-1} (1+t) N \leq C (1+t^3) N^{-2/3}.
\end{equation}
Finally, \eqref{nablaPhi} gives
\begin{equation}
\norm[\infty]{\nabla\Phi^t}^4 \leq Ct^4 N^{-4/3}.
\end{equation}
Thus,
\begin{equation}
R_2(t) \leq \alphaeff(t) + (1+t^4)N^{-2/3}.
\end{equation}

\noindent\textbf{Gronwall Estimate.} In total, we have found
\begin{equation}
\partial_t \alphaeff(t) \leq C \alphaeff(t) + C (1+t^4)N^{-2/3}.
\end{equation}
From the Gronwall estimate \eqref{gronwall}, we conclude that
\begin{equation}
\alphaeff(t) \leq e^{Ct} \alphaeff(0) + C e^{Ct} N^{-2/3}.
\end{equation}
From \eqref{alpha_gamma}, the desired estimate \eqref{main_result1} follows.
\end{proof}

\begin{proof}[Proof of Theorem~\ref{thm:main_result2}]
We would like to control $\alpha(t) := \alpha(\psi^t,\{ \varphi_j^t \}_{j=1,\ldots,N})$, where $\psi^t$ is the solution to the Schr\"odinger equation \eqref{Schr_main}, and $\varphi_j^t$ are the solutions to the fermionic Hartree equations \eqref{Hartree_main}. We can now directly use \cite[Lemma~7.4 for $\gamma=1$]{petrat:2015}, i.e.,
\begin{equation}
\partial_t \alpha(t) \leq C N^{-1/6} \sqrt{ \norm[\infty]{v^2 * \rho^t}} \, \Big( \alpha(t) + N^{-1} \Big).
\end{equation}
With Lemma~\ref{lem:v2_hartree}, we then have
\begin{equation}
\partial_t \alpha(t) \leq C e^{Ct} \Big( \alpha(t) + N^{-1} \Big).
\end{equation}
Applying again the Gronwall estimate \eqref{gronwall} and \eqref{alpha_gamma} gives the desired result \eqref{main_result2}.
\end{proof}

\absatz

\noindent{\it Acknowledgments.} I would like to thank Maximilian Jeblick, David Mitrouskas, Peter Pickl, and Robert Seiringer for many helpful discussions, and I am particularly grateful to the referees for valuable suggestions and comments. The research leading to this work has received funding from the People Programme (Marie Curie Actions) of the European Union's Seventh Framework Programme (FP7/2007-2013) under REA grant agreement n\textdegree~291734. My research is supported by a Postdoc scholarship of the German Academic Exchange Service (DAAD).

\bibliographystyle{plain}
\bibliography{references}

\end{document}